\documentclass{article}
\usepackage{graphicx}
\usepackage{latexsym}
\usepackage{amsmath}
\usepackage{amsfonts}
\usepackage{comment}
\usepackage{float}
\usepackage{amssymb}
\usepackage{calrsfs}
\usepackage[left=2cm,right=2cm,top=2cm,bottom=2cm]{geometry}
\usepackage[mathscr]{euscript}
\usepackage{subcaption}
\usepackage{mathtools}
\usepackage{hyperref}
\usepackage[utf8]{inputenc}
\usepackage[T1]{fontenc}
\usepackage{euscript,fullpage,amsmath,amssymb}
\usepackage[T1]{fontenc}
\usepackage{enumerate}
\usepackage{graphicx}
\usepackage{float}
\usepackage{setspace}
\usepackage{amssymb}
\usepackage{bbm}
\usepackage{subcaption}
\usepackage{stmaryrd}
\usepackage{amsmath}
\usepackage[T2A,T1]{fontenc}
\usepackage[russian,english]{babel}
\newtheorem{definition}{Definition}
\newtheorem{remark}{Remark}
\newtheorem{theorem}{Theorem}
\newtheorem{result}{Result}
\newtheorem{proposition}{Proposition}
\newenvironment{proof}[1][Proof]{\noindent\textbf{#1.} }{\ \rule{0.5em}{0.5em}}
\newtheorem{example}{Example}

\newcommand{\eeq}{\end{equation}}
\newcommand{\beq}{\begin{equation}}
\newcommand{\nuq}[1]{\label{#1} \eeq}

\begin{document}
\title{Matching of observations of dynamical systems, with applications to sequence matching}
\author{Th\'eophile Caby,\\
CMUP, Departamento de Matemática, Faculdade de Ciências, Universidade do Porto,\\
Rua do Campo Alegre s/n, 4169007 Porto, Portugal.\\
 caby.theo@gmail.com}
\date{}

\maketitle
\begin{abstract}
We study the statistical distribution of the closest encounter between generic smooth observations computed along
different trajectories of a rapidly mixing dynamical system. At the limit of large trajectories, we obtain a distribution of Gumbel type that depends on both the
length of the trajectories and on the Generalized Dimensions of the image measure. It is also modulated
by an Extremal Index, that informs on the tendency of nearby observations to diverge along with the evolution of the dynamics. We give a formula
for this quantity for a class of chaotic maps of the interval and regular observations. We present diverse
numerical applications illustrating the theory and discuss the implications of these results for the study
of physical systems. Finally, we discuss the connection between this problem and the problem of the longest matching block common to different sequences of symbols. In particular, we obtain a distributional result for strongly mixing processes.
\end{abstract}

\section{Introduction}

Certain real-world systems, such as climate, take place in high-dimensional spaces and exhibit chaotic and multi-scaled properties. To study such complex dynamics, physicists often have access to only a limited number of observable quantities through the measurement process. The latter can be modeled by computing an observation function along a typical trajectory of the system. Understanding the geometric and statistical properties of such observations, and their relationship to the properties of the original underlying system is a problem of great interest in physics, that has been instigated only recently. The study of the recurrence properties of observations have been initiated by Rousseau and Saussol in \cite{roussol,jerobs}, in which asymptotic and distributional results were obtained for both hitting times and return times of observations. In a recent paper \cite{obsrec}, this problem was studied from the point of view of Extreme Value Theory (EVT). This approach turned out to provide information on the local geometry of the image measure, which, for a good choice of observation, can characterize the local fractal structure of the original underlying attractor. In this paper, we pursue the statistical analysis of observed dynamical systems by studying the statistics of the shortest distance between several observed trajectories. Closely related problems have gained interest in recent years. The case of real, unobserved trajectories was considered in \cite{d2} and \cite{dq}, using EVT techniques, while asymptotic results for the shortest distance between two orbits were obtained in \cite{short} and then generalized to multiple orbits \cite{mult} and finally to observed orbits \cite{encoded}.

Yet another motivation to study this problem is its deep relationship with a seemingly distinct one; the length of the longest matching block common to different sequences of symbols drawn from the same probability distribution. This old problem has been initiated by Waterman and Arratia, who brought a plethora of results in the i.i.d. case \cite{erdos,extreme}, most of which are presented in the reference book \cite{watbook}. Several authors have extended these results, giving for example distributional results in the i.i.d case \cite{distrib,house}. In many applications,  however, the sequences cannot be modeled as i.i.d. sequences. For example, in biological applications, genes are specific sequences encoding information, and DNA brands do not constitute independent sequences of nucleotides. When it comes to written text, a complex dependence structure can arise from specific sequences of letters, such as words, and higher-order syntactic and narrative structures. Recently, Barros, Liao and Rousseau adopted a dynamical system point of view to give the asymptotic behavior of the length of the longest sub-sequence common to different $\alpha-$mixing sequences \cite{short,mult}. This problem is different from the present one, because the sub-sequences may be present at different locations of the different strings of symbols, but we will also follow a dynamical system approach to derive our results. 

The paper is organized as follows: In the first section, we present the problem and derive our main result concerning the convergence of the statistics of observation matching to a Gumbel distribution. In the following sections, we discuss the parameters of the limit law, since these quantities provide relevant dynamical information on the system and can be estimated numerically. We first focus on the generalized dimensions of the image measure by emphasizing their central role in the statistical properties of observations. We also study their relations with the generalized dimension spectrum of the original measure. In the third section, we derive a formula for the Extremal Index associated with this problem for a class of chaotic maps of the interval and perform a numerical study of this index for higher dimensional systems. In the last part, we present some applications of our results to sequence matching problems. In particular, we obtain distributional results for the length of the longest sequence of symbols common to independent strings of symbols drawn from the same strongly mixing probability measure.\\

\section{The general approach}

Let us consider the dynamical system $(\mathcal{M},T,\mu)$, where $\mathcal{M}$ denotes the phase space and $T:\mathcal{M} \to \mathcal{M}$ is a discrete transformation \footnote{it could be a discretized version of a flow} that leaves the probability measure $\mu$ invariant. In order to model the process of measurement, we consider a $C^1$ function $f:\mathcal{M} \to \mathcal{J}$, which we refer to as the {\em observation}. Both the phase space $\mathcal{M}$ and the observation space $\mathcal{J}$ are compact metric spaces endowed with two distances that we will both call $d$ to simplify notations. For physical applications, we take $\mathcal{J} \subset\mathbb{R}^m$, as observational data usually consists of a collection of real numbers that can be arranged into vectors. For applications to the problem of sequence matching, we will take $\mathcal{M}$ to be the space of all infinite sequences of symbols of a given alphabet $\mathcal{A}$.  Because we are interested in the statistical properties of observations, we need a measure that is supported in the observational space.

\begin{definition} We call the push-forward, or image of the measure $\mu$ by the function $f$, the measure $\mu_f$ defined by $$\mu_{f}(A)=\mu(f^{-1}(A)),$$ for all $A \subset \mathcal{J}$ such that $f^{-1}(A)$ is $\mu-$measurable. \end{definition}

A more detailed presentation of this object is available in \cite{roussol} and a discussion of its properties can be found in \cite{obsrec}. 

\begin{definition}\label{defdq}
We call the generalized dimension of order $q\neq 1$ of the image measure $\mu_f$ the following quantity (if it exists):

\begin{equation}
D_q^f=\underset{r\to 0}{\lim} \frac{\log \int_{\mathcal{J}} \mu_f(B(y,r))^{q-1}d\mu_f(y)}{(q-1)\log r}. 
\end{equation}

$B(y,r)$ denotes a ball centered at $y\in \mathcal{J}$ of radius $r$.\\

The information dimension of $\mu_f$ is defined as
\begin{equation} 
D_1^f=\underset{q\to 1}\lim D_q^f.
 \end{equation}
 We write $D_q=D_q^{Id}$, the generalized dimension of order $q$ of the original measure.
\end{definition}

We will place ourselves in physical situations where the limits defining the previous quantities exist. Now that we have introduced the important objects of the theory, we go forward and consider the following process:

$$ Y_i=-\log\max_{j=2,\dots,q}d(f(T^ix_1),f(T^ix_j)),$$

$(x_1,...,x_q) \in \mathcal{M}^q$ being a starting point drawn from the product measure $\mu_q$ with support in $\mathcal{M}^q$. To follow the usual procedure of Extreme Value Theory, we consider a sequence of thresholds $u_n(s)$, where $s\in \mathbb{R}$, such that:

\begin{equation}\label{tau} \mu_q(Y_0 > u_n(s)) \sim \frac{e^{-s}}{n}. \end{equation}

Since the $q$ trajectories are independent, we also have:

\begin{equation}\label{dqf} \begin{aligned} \mu_q(Y_0> u_n(s)) &=\int_\mathcal{J} \mu_f (B(y,e^{-u_n}))^{q-1}d \mu_f(y)\\ & \sim e^{-u_nD^f_q(q-1)}, \end{aligned} \end{equation}

from definition 2. To satisfy both scalings \ref{tau} and \ref{dqf}, we take

$$u_n(s)=\frac{\log n}{D_q^f(q-1)}+\frac{s}{D_q^f(q-1)}.$$

Now, for a given threshold $u_n$, the quantity $\mu_q(Y_0> u_n)$ gives the probability of having all the observations contained in the same small region of the observational space; a ball of radius $e^{-u_n}$ centered at $f(x_1)$. Equivalently, it gives the probability that the product dynamics has entered the following target set: 

$$S^q_n=\{(s_1,...,s_q)\in \mathcal{M}^q, \max_{j=2,\dots,q}d(f(s_1),f(s_j)) < e^{-u_n}\}.$$

Following the ideas of \cite{fft}, studying the behavior of the maximum of the process $(Y_i)$ over a trajectory of size n:

$$M_{n,q}(x_1,...,x_q)=\max \{Y_0,\dots,Y_{n-1}\},$$

and in particular its cumulative distribution:

$${F_n}(u_n) = \mu_q(\{(x_1,...,x_q) \in \mathcal{M}^q \mbox{ s.t. } M_{n,q}(x_1,...,x_q) \leq u_n \}),$$

is equivalent to studying the Hitting Time Statistics of the product dynamics in the set $S^q_n$. Indeed, ${F_n}(u_n)$ gives the probability that the dynamics has not entered the set $S^q_n$ after $n$ iterations of the dynamics. We can now apply results from EVT, in particular, the spectral theory developed by Keller and Liverani \cite{kl,k}, to obtain the convergence of ${F_n}(u_n)$ to its limit law.

\begin{proposition} 
For a large class of exponentially-mixing systems and {\em regular} observations, there exists $0<\theta_q^f\leq 1$ such that:

\begin{equation}
|F_n(u_n(s)) -\exp(-\theta^f_q e^{-s})| \underset{n\to\infty}\to 0.
\end{equation}

\end{proposition} 

The term $\theta^f_q$ is called the Extremal Index (EI) and quantifies the tendency of the process $(Y_i)$ to form clusters of high values. To be applicable, the spectral theory requires that the couple system/observation satisfies the so-called REPFO property \cite{kl,k}, which is verified for rapidly mixing systems for which the measure of the nested target sets $S^q_n$ goes to zero in a regular fashion. More detailed presentations of the theory and its domain of application can be found in various publications \cite{d2,dq,kl,k,synchro}. The theory is proven to be particularly adapted to expanding maps of the interval \cite{book,synchro} and certain well-behaved 2-dimensional systems \cite{bakersandro}.

 More classical tools can also be used to prove the convergence to the limit law, in particular under the following conditions, that are particularly adapted to processes generated by dynamical systems.
 
 \begin{definition}\label{defd}
 
 We say that the condition \foreignlanguage{russian}{Д}$_1(u_n)$ is satisfied for the process $Y_0,Y_1,...$ if there exist a function $\gamma(n,t)$ such that for every $l,t,n\in \mathbb{N}$,
 
\begin{equation}\label{ffcond}
 |\mu_q(A_n\cap B_{t,l,n})-\mu_q(A_n)\mu_q(B_{0,l,n})| \le \gamma(n,t),
 \end{equation}
 
 where $A_n=\{Y_0>u_n,Y_1\le u_n\}$, $B_{t,l,n}=\bigcap_{i=t}^{t+l-1}T^{-i}(A_n^c)$, and the function $\gamma(n, t)$ is such that it is decreasing in $t$ for each $n$ and such that there exists a sequence $(t_n)_n\in\mathbb{N}$
satisfying $t_n = o(n)$ and $n\gamma(n, t_n)  \underset{n\to\infty}\to 0$.\\
 \end{definition}
 \begin{definition}\label{defdkun}
 
 We say that \foreignlanguage{russian}{Д}$_1'(u_n)$ holds if there exist a sequence $(k_n)_n$ such that:
 
 \begin{enumerate}
\item  $k_n \underset{n\to\infty}\to \infty$.
\item $k_nt_n=o(n)$, where $(t_n)_n$ is the sequence in definition \ref{defd}.
\item $\underset{n\to\infty}\lim n\sum_{j=2}^{\lfloor\frac{n}{k_n}\rfloor-1}\mu_q(Y_0>u_n\cap Y_1\le u_n\cap Y_{j}>u_n)=0.$\label{dkun}
\end{enumerate}
\end{definition}

Under these two conditions, the result of Proposition 1 holds \cite{book}. We stress that these conditions depend both on the application $T$ and on the observation $f$. \foreignlanguage{russian}{Д}$_1(u_n)$ is expected to hold for rapidly mixing systems and regular observations. In particular, we show in the annex that, at least in the context of symbolic dynamics and if $f=Id$, strong exponential mixing implies \foreignlanguage{russian}{Д}$_1(u_n)$. Condition \foreignlanguage{russian}{Д}$_1'(u_n)$ concerns the clustering structure of the process $Y_i$. More particularly, it controls the probabilities of short returns to the target set $S^q_n$. It is not our focus to give more appropriate conditions of convergence to the limit law, since these can be hard to check in dimension more than one, or sometimes two\footnote{for simple systems such as automorphisms of the torus \cite{jorgeauto} or systems admitting a product structure \cite{jorgemult}}, and even more so when a non-trivial observation $f$ is introduced. We will however provide numerical evidence of the convergence to the extreme value law. Let us now discuss the values of the different parameters of the limit law, that can acquire a physical interpretation. 

\section{The Generalized Dimensions of the image measure $D_q^f$}
%

%

\subsection{On the relation between the Generalized dimensions of the image measure and the one of the original invariant measure}

We have seen in the preceding section that the quantity $D_q^f$ appears as a parameter of the limit law and therefore modulates the synchronization properties of observations. In fact, these quantities play a central role in different aspects of the statistical properties of observations, and in particular their recurrence times. It is well known that both return and hitting times of certain chaotic systems in small balls (in fact, re-scaled versions of these quantities) have large deviations that are governed by the spectrum of generalized dimensions $D_q$ of the invariant measure \cite{dq,ldr}. These kinds of large deviations relations are known to hold for real trajectories, but similar results are also expected to apply to the recurrence times of observations for such systems. This matter will be investigated more in detail in a future publication. For now, let us focus on the properties of $D_q^f$, and in particular on their relation to the generalized dimensions $D_q$ of the original system.
In \cite{hk}, Hunt and Kaloshin give results concerning the effect of typical projections on the generalized dimensions for $1\le q \le 2$. In this range, they show that if $\mathcal{M}$ is a compact subset of $\mathbb{R}^n$ and $\mathcal{J}=\mathbb{R}^m$, and if the generalized dimension of order $q$, $D_q$ of the invariant measure exists, then: 

\begin{equation}\label{hk}
D_q^f=\min(D_q,m),
\end{equation}

 for a prevalent set of $C^1$ observables. See \cite{prev} for a review of prevalence, which is a notion of genericity for infinite-dimensional spaces. For $q>2$, no such result holds and the behavior of $D_q^f$ in this range is not yet completely understood. Under the light of  Hunt and Kaloshin's result, it is possible to access the correlation dimension $D_2$ of a physical system using a generic observation if the rank is large enough (larger than the correlation dimension of the original attractor). This quantity can be obtained by fitting the empirical distribution of $M_{n,q}$ and extracting the desired parameter, as we will do in the following subsection. Such EVT-based methods of computation of fractal dimensions is in use in climate studies, in particular for the computation of the local dimension, which can be used as a tool to characterize certain climatic patterns \cite{nature,messori}.\\
 
  Different kinds of large rank observations can be used by physicists to recover information on the original system. A first approach is to record simultaneously the value of a scalar quantity at different locations of a spatially extended system. These measurements can be arranged into a vector and constitute a so-called gridded observation in $\mathbb{R}^m$. Instead of recording the same quantity at different positions, one can also record different independent observables (temperature, position, speed, pressure, ...) at a given time. Yet another possibility is to consider delay coordinates observables used in embedding techniques \cite{takens}. In this context, it is well known that if one considers enough delay coordinates (larger than $\lceil2D_0\rceil$), the dynamics of the observation settles on an object (the so-called reconstructed attractor) that is a smooth deformation of the original attractor, which preserves the dimensions  \cite{takens}. With our approach, only $m\ge D_2$ delay coordinates are required to access the correlation dimension $D_2$, although the reconstructed attractor is now likely to have a different fine structure from the original one.\\

\subsection{Numerical extraction of $D_q^f$}

  Let us now investigate the values of $D_q^f$ for $q>2$ from a numerical perspective. This procedure will also allow us to experimentally intuit the convergence of the distribution of $M_{n,q}$ to its limit law. Let us consider a system for which the explicit values of $D_q$ are available; the motion on a Sierpinski gasket given by the following Iterated Function System on the unit square $\mathcal{M}=[0,1]^2$:
\beq
 \left\{
    \begin{array}{ll}
   T_1(x,y)=(x/2,(y+1)/2), \;p_1=1/4,\\
T_2(x,y)=((x+1)/2,(y+1)/2),\; p_2=1/4,\\
T_3(x,y)=(x/2,y/2),\; p_3=1/2.\\
    \end{array}
\right.
\nuq{17.5}
At each iteration, the application $T_i$ is applied with probability $p_i$. The associated generalized dimensions spectrum is given, for $q\neq 1$, by \cite{dq}:
 
\begin{equation}\label{serp}
D_q=\frac{\log_2(p_1^q+p_2^q+p_3^q)}{1-q}.
\end{equation}  

  In figure (\ref{dqs}), we compare the numerical estimates of $D^f_q$ for different observations $f$ and the theoretical values of $D_q$ given by equation (\ref{serp}). These estimates are obtained by evaluating the scale parameter of the empirical maximum distribution of the process $(Y_i)$ over blocks of size $5.10^4$, using the maximum likelihood estimator provided by the Matlab function gevfit \cite{gevfit}. The results are averaged over 10 runs, using each time different randomly selected trajectories of length $2.10^8$. The error bars represent the standard deviations of the results over these 10 runs.\\
   The functions $f_1$, $f_2$ are diffeomorphisms, which are known to preserve the generalized dimensions \cite{hk}. Indeed, for these two functions, good agreement is found, so that the two curves are hardly distinguishable visually in the picture. These results suggest that this method of computation of $D_q$ can be completed and even improved by introducing a diffeomorphism computed along the orbits of the system, which may, if well chosen, speed up the convergence of the method and provide better estimates. Function $f_3$ is a very oscillatory function, which gives a point in the observational space many antecedents, having the effect to alter significantly the fine structure of the image measure. We do not know whether the disagreement with the $D_q$ spectrum is due to the method not being at convergence, or if is a sign that the spectrum is not preserved under the action of $f_3$. However, the small disagreement for $q=2$ seems to indicate that the method may not be at convergence, since the correlation dimension is preserved by typical observations. $f_4$ is not a diffeomorphism either, but has a more simple structure. For this function, the generalized dimensions seem to be preserved. $f_5$ is a degenerate function yielding values close to 1.\\
 \begin{figure}[h!]
 \centering
\includegraphics[height=3.5in]{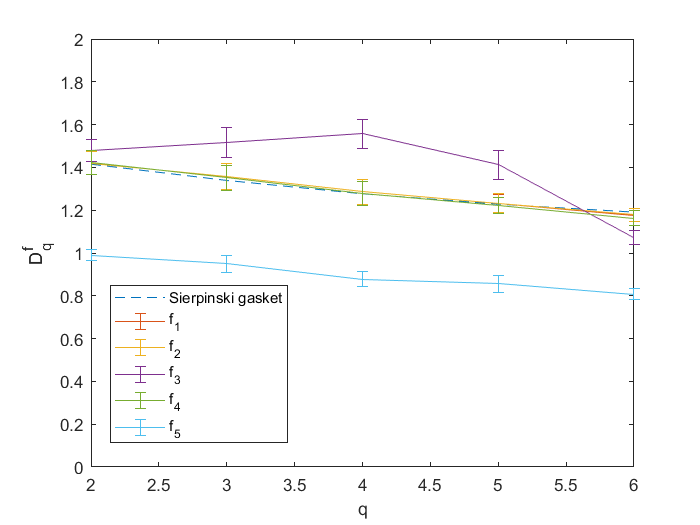}
\caption{Numerical estimates of $D_q^f$ for different observations: $f_1=Id$, $f_2(x,y)=(2x+y,2y)$, $f_3(x,y)=(\sin(\frac1x),\cos(\frac1y))$, $f_4(x,y)=((x-0.5)^2,2y)$ and $f_5=(1,y^2+x)$. In dashed lines is the $D_q$ spectrum of the underlying system. Estimates are computed as described in the text.}
\label{dqs}
\end{figure}

In \cite{obsrec}, we showed that for the two-dimensional baker's map, which has a non trivial $D_q$ spectrum, a typical linear uni-dimensional projection gives $D_q^f=1$ for all $q$. Overall, this result, along with our numerical computations, suggests that Hunt and Kaloshin's results may extend to $q>2$ for a certain class of measures and certain smooth observations. We hope to provide their characterization on future investigations.
 
\section{The Extremal Index $\theta^f_q$}
\subsection{An explicit formula for expanding maps of the interval}
When considering real trajectories (i.e. when $f=Id$), the Extremal Index $\theta_q$, and more specifically the quantity $$h_q=\frac{\log(1-\theta_q)}{1-q},$$ encodes the hyperbolic properties of the system (see \cite{dq} for a detailed review). In particular, $h_q$ as a function of $q$ is constant for maps with constant Jacobian and is close to the metric entropy of the system (its Lyapunov exponent in dimension 1). When an observation $f$ is introduced, the use of the Extremal Index to quantify the rate at which nearby trajectories diverge becomes less relevant, in particular because two nearby points in observational space may have antecedents far away in the actual phase space of the system. Let us investigate this matter with more detail.
Keller and Liverani \cite{kl} provide a general formula for the Extremal Index of time series originated by dynamical systems. Applied to the present situation and if the limits defining the different quantities exist, we have that

\begin{equation}\label{deftheta}
\theta^f_q=1-\sum_{k=0}^{\infty}p_{k,q},
\end{equation}

where 

\begin{equation}\label{p0}
p_{0,q}=\lim_{n\to\infty}\frac{\mu_q(S^q_n \cap T^{-1} S^q_n)}{\mu_q(S^q_n)}
\end{equation}

and for $k\ge 1$,

\begin{equation}\label{pk}
p_{k,q}=\lim_{n\to\infty}\frac{\mu_q(S^q_n \cap \bigcap_{i=1}^k T^{-i}(S^q_n)^c \cap T^{-k-1} S^q_n)}{\mu_q(S^q_n)}.
\end{equation}

In this general set up, obtaining a formula for $\theta^f_q$ is challenging, so let us place ourselves in the more simple case of expanding maps of the unit interval $I=[0,1]$. We define the following sets for a given $x\in I$ :
$$A_0(x)=\{y \in I \text{ such that } f(y)=f(x) \text{ and } f(Ty)=f(Tx)\}$$
and
$$A_k(x)=\{y \in I \text{ such that } f(y)=f(x), f(T^iy)\neq f(T^ix),\text{ for } i=1,..,k\text{ and } f(T^{k+1}(y))=f(T^{k+1}(x))\}.$$

\begin{proposition}
Let $T$ be an expanding map of the unit interval $I=[0,1]$ which is $C^1$ by part and admitting an absolutely continuous invariant measure $d\mu(x)=h(x)dx$. Let $f : I \to J \subset \mathbb{R}$ be $C^1$ by part, finite to one and such that $f' \neq 0$ on $I$. Suppose moreover that the couple $(T,f)$ satisfies the conditions of Proposition 1, that

\begin{equation}\label{h1}
\mu(\{x \in I,A_0(x)=\{x\} \})=1
\end{equation} 

and that, for all $k\ge 1$,

\begin{equation}\label{h2}
\mu(\{x\in I, A_k(x)= \emptyset \})=1.
\end{equation}

Then:
 
\begin{equation}\label{thetaq}
\theta^f_q=1-\frac{\int_I \frac{h(x)^q}{\max(|f'(x)|,|(f\circ T)'(x)|)^{q-1}}dx}{\int_I \sum_{(y_1,...y_{q-1})\in (f^{-1}\{f(x)\})^{q-1}} \prod_{i=1}^{q-1}\frac{h(y_i)}{|f'(y_i)|}h(x)dx}.
\end{equation}
\end{proposition}

\begin{proof}
We write the proof for $q=2$, the cases $q>2$ can be obtained in a similar fashion.
We start from formula (\ref{deftheta}) and evaluate both the numerators and the denominators defining the $p_{k,2}$ terms. Let us start by the denominator, for the case $k=0$. Following the lines of the proof in \cite{synchro} (where the case $f=Id$ is treated), and making use of the mean value theorem, we get:
\begin{equation}\label{1}
\begin{aligned}
\mu_2(S^2_n)&\sim \int_I \sum_{y\in f^{-1}\{f(x)\}} \mu(B(y,\frac{e^{-u_n}}{|f'(y)|})) d\mu(x)\\
                     &\sim 2e^{-u_n} \int_I \sum_{y\in f^{-1}\{f(x)\}} \frac{h(y)}{|f'(y)|} h(x)dx.
\end{aligned}
\end{equation}

On the other hand, still for the case $k=0$, we get for the numerator:

\begin{equation}\label{2}
\begin{aligned}
\mu_2(S^2_n \cap T^{-1} S^2_n) &\sim \int_I \sum_{y\in A_0(x)} \mu(\{z\in I, z\in B(y,\frac{e^{-u_n}}{|f'(y)|})\cap Tz \in B(Ty,\frac{e^{-u_n}}{|f'(Ty)|}\})d\mu(x)\\
                                                   &\sim \int_I \sum_{y\in A_0(x)} \mu(\{z\in I, |z-y| \le \frac{e^{-u_n}}{|f'(y)|} \cap T'(y)|y-z| \le \frac{e^{-u_n}}{|f'(Ty)|}\}) h(x)dx.\\
                                                   &= \int_I \sum_{y\in A_0(x)} \mu(\{z\in I, |z-y| \le  \min(\frac{e^{-u_n}}{|f'(y)|},\frac{e^{-u_n}}{|T'(y)f'(Ty)|})\}) h(x)dx.\\
                                                   &\sim  2e^{-u_n} \int_I \sum_{y\in A_0(x)} \frac{h(y)h(x)}{\max(|f'(y)|,|(f\circ T)'(y)|)} dx.\\
\end{aligned}
\end{equation}

By a similar reasoning, we get that for $k \ge 1$,

\begin{equation}\label{3}
\mu_2(S^2_n \cap \bigcap_{i=1}^k T^{-i}(S^2_n)^c \cap T^{-k-1} S^2_n) \sim 2e^{-u_n} \int_I \sum_{y\in A_k(x)} \frac{h(x)h(y)}{\max(|f'(x)|,|(f(T^{k+1}(y))'|)}dx.\\
\end{equation}

Finally, combining eqs. (\ref{deftheta}),(\ref{1}), (\ref{2}) and (\ref{3}), we obtain

\begin{equation}
\theta^f_2=1 - \sum_{k=0}^{+\infty} \frac{\int_I \sum_{y\in A_k(x)} \frac{h(x)h(y)}{\max(|f'(x)|,|(f\circ T^{k+1})'(y)|)}dx}{\int_I \sum_{y\in f^{-1}\{f(x)\}} \frac{h'(y)}{|f'(y)|} h(x)dx}.
\end{equation}

This formula is still difficult to handle, but under condition (\ref{h2}), $p_{k,2}=0$ for $k>0$, and if condition (\ref{h1}) holds, we obtain 

\begin{equation}\label{thetafin}
\begin{aligned}
\theta^f_2&=1-p_{0,2}\\
             &=1-\frac{\int_I \frac{h(x)^2}{\max(|f'(x)|,|(f\circ T)'(x)|)} dx}{\int_I \sum_{y\in f^{-1}\{f(x)\}} \frac{h(y)h(x)}{|f'(y)|}dx}.
\end{aligned}
\end{equation}

We can generalize this result for $q\ge 2$ to obtain the desired result.
\end{proof}\\

\begin{remark}
For a given map $T$, assumptions (\ref{h1}) and (\ref{h2}) should be satisfied for a generic observation $f$. The cases where these assumptions are not satisfied are when $T$ and $f$ share some particular symmetries and similarities in their structures. For example, $\mu(A_0(x) = \{x\})\neq 1$ if both the graphs of $T$ and $f$ are symmetric with respect to the straight line of equation $x=1/2$.
\end{remark}


\begin{example} Let us take $Tx=2x \mod 1 $ and $$f(x)=\left\{
    \begin{array}{ll}
        2x & \mbox{ if } 0\le x\le 1/2 \\
        3/2-x & \mbox{ if } 1/2<x\le 1.\\
    \end{array}
\right.$$

$T$ is strongly mixing and the couple $(T,f)$ satisfies conditions (\ref{h1}) and (\ref{h2}), so that $(T,f)$ should satisfy the conditions of existence of the limit law, \foreignlanguage{russian}{Д}$_1(u_n)$ and \foreignlanguage{russian}{Д}$_1'(u_n)$. It constitutes a good test for our results, since computations can be worked out quite easily. Applying formula (\ref{thetaq}), we get $$\theta^f_q=1-p_{0,q}=1-\frac{2+2^{2-q}}{1+3^q}.$$ This result is confirmed by numerical experiments (see figure (\ref{1})). We used the estimator $\hat{\theta}_{5}$ introduced in \cite{ei}, which consists in evaluating the 5 first $p_{k,q}$ terms appearing in formula (\ref{pk}). To do so, we compute Birkhoff sums for both the numerator and the denominator defining the $p_{k,q}$ terms. It requires fixing a high threshold $u$, that we take here equal to the 0.99999-quantile of the empirical $Y_i$ distribution. As expected, we find that all the $p_{k,q}$ are 0 or very close to 0 for $k\ge1$. The results are averaged over 10 runs, with trajectories of length $2.10^7$. The error bars in figure (\ref{1}) represent the standard deviations of the results over these 10 runs.
In this example, the exact limit distribution can be computed explicitly; since the image measure is absolutely continuous with a density that does not vanish and that admits no singularities, $D_q^f=1$ for all $q$.\\

\end{example}

\begin{figure}[h!]
\centering
\includegraphics[height=3.5in]{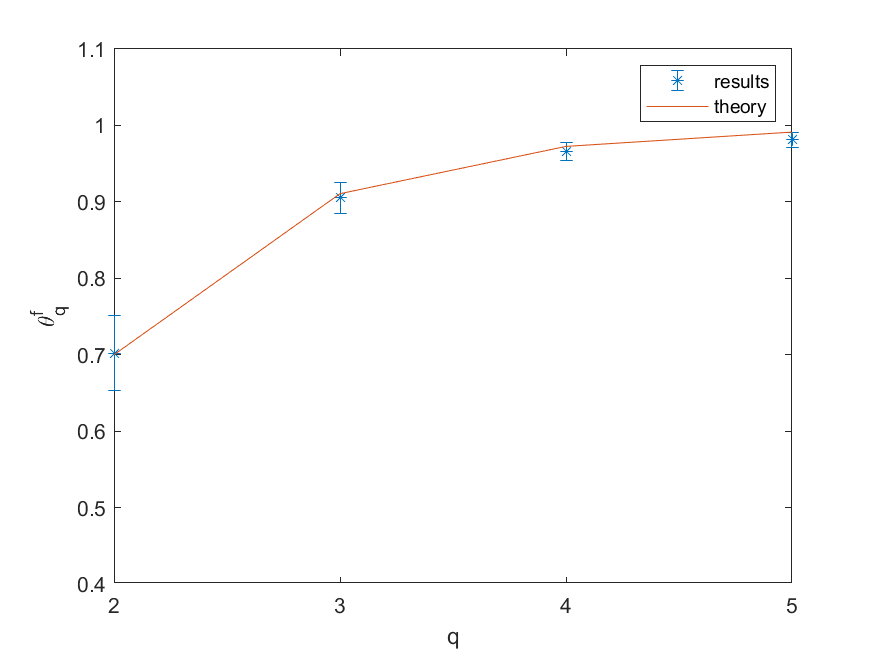}
\caption{Comparison between theory and computation for the $\theta_q^f$ spectrum of the system in Example 1. Details of the computation can be found in the text.}
\label{1}
\end{figure}

\subsection{Numerical estimation of $\theta_q^f$ in higher dimensional systems}

A general formula for higher-dimensional system is out of scope, but we expect that with conditions of `non compatibility' between the dynamics and the observation analogue to conditions (\ref{h1}) and (\ref{h2}), all the $p_{k,q}$ terms are 0 for $k\ge1$. The aim of this section is to show that this hypothesis is corroborated by numerical experiments.\\

 For the uni-dimensional case, the presence of the derivative of the observation in formula (\ref{thetaq}) renders the interpretation of $\theta_q^f$ less apparent than in the case $f=Id$. However, we point out two facts :

\begin{itemize}
\item For a given observation $f$, the larger the values of $|T'|$ over phase space, the larger the values of $\theta_q^f$, so this index can still quantify the hyperbolic properties of $T$.
\item For a given map $T$, the more the points in the observational space have antecedents by $f$, the larger is the denominator in equation (\ref{thetaq}), and the larger is $\theta_q^f$. Oscillatory observations yield higher values for the extremal index.

\end{itemize}

We expect analogous properties to hold for higher dimensional systems. To test this statement, we compare in figure (\ref{thetf}a) the estimates of $\theta_q^f$ for the 2-dimensional H\'enon system, defined by $T(x,y)=(1-ax^2+y,bx)$, with $a=1.4$ and different values of $b$ such that the system admits a strange attractor \cite{henon}. We consider the  observation $f(x,y)=\frac{x+y}{2}$. The determinant of the Jacobian is given by $b$. We find indeed that for this fixed choice of observation, the more the original system tends to separate trajectories (the higher is parameter $b$), the higher are the values of $\theta_q^f$, even for uni-dimensional projections. The estimates $\hat{p}_{k,q}$ of the $p_{k,q}$ terms, for $k>0$ are all null or close to 0 for all the observations that we considered, as conjectured earlier.\\

In figure (\ref{thetf}b), we plot the estimates of the extremal index for 2-dimensional H\'enon system (using the usual parameters $a=1.4$, $b=0.3$) and different observations. We observe that for one-to-one observations, ($f_1$, $f_2$ and $f_3$), the $\theta_q^f$ spectrum remains relatively low, although the form of the Jacobian can impact significantly the values of $\theta_q^f$. When the observation ceases to be one-to-one, the whole spectrum of extremal indices increases significantly (see the curve for $f_4$). This effect is even more important for the very oscillatory function $f_5$. For analogous reasons, we expect that for high dimensional systems, observations that perform a large drop of dimensionality tend to yield higher values for the $\theta_q^f$ spectrum.

%

\begin{figure}
\centering
\begin{subfigure}{.5\textwidth} \label{thet2}
  \centering
  \includegraphics[width=1\linewidth]{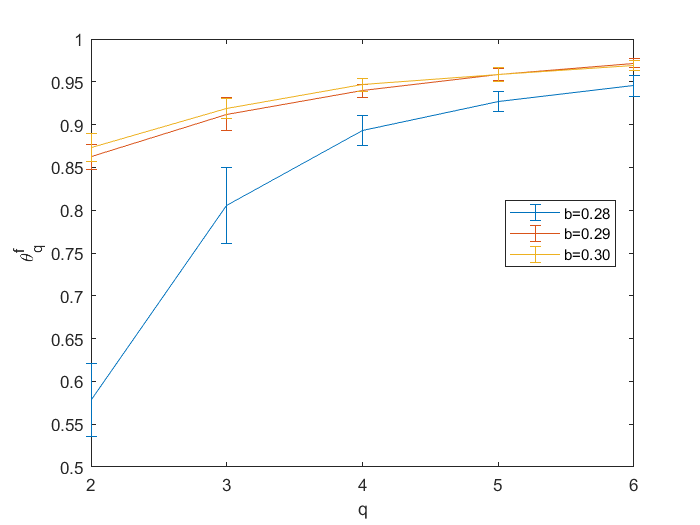}
  \caption{ }
\end{subfigure}%
\begin{subfigure}{.5\textwidth} \label{thetf}
  \centering
  \includegraphics[width=1\linewidth]{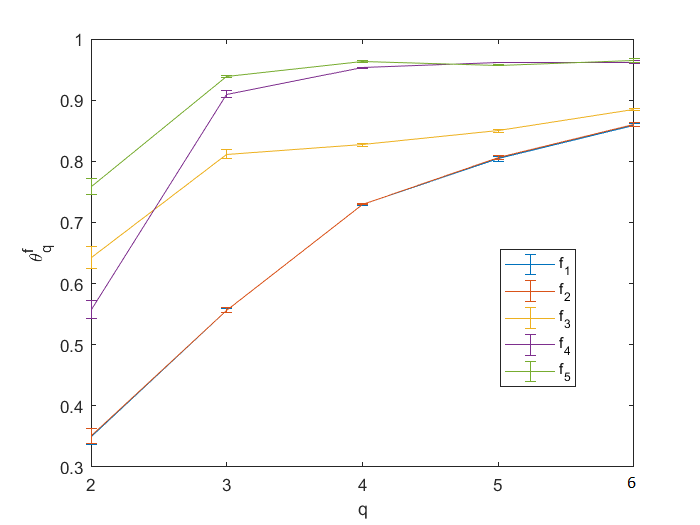}
  \caption{ }
\end{subfigure}
\caption{Left: Estimates for the $\theta^f_q$ spectrum computed for a H\'enon system with different parameters $b$ and for the observation $f(x,y)=\frac{x+y}{2}$. Right: Estimates for the $\theta^f_q$ spectrum computed for the H\'enon system (b=0.3) and different observations : $f_1=Id$, $f_2(x,y)=(100x+y,100y)$, $f_3(x,y)=(x,100y)$, $f_4(x,y)=(x^2,y^2)$, $f_5(x,y)=(\sin(1/x),\cos(1/y))$. For both figures, we used the estimate $\hat{\theta}_5$ introduced in \cite{ei}, with trajectories of length $10^6$ and a threshold value equal to the $0.999$ quantile of the empirical $Y_i$ distribution. The error bars represent the standard deviation of the results over 10 runs.}
\label{thetf}
\end{figure}

%

\section{Application to Sequence Matching}

In this section, we discuss the connection between the present problem and sequence matching problems. Let

\beq
 \left\{
    \begin{array}{ll}
  X^1=X_1^1X_2^1\hdots X_n^1,\\
  \vdots\\
  X^q=X_1^qX_2^q\hdots X_n^q\\
    \end{array}
\right.
\nuq{18.5} 

be $q$ sequences of symbols of length $n$, drawn from the finite alphabet $\mathcal{A}$ with the same probability distribution $\mathbb{P}$. We will denote $\bar{X}_i=(X_i^1,X_i^2,...,X_i^q)$. We suppose that the sequences have a good dependence structure that we will describe later. We are interested in deriving the limit distribution of the length of the longest matching block for the $q$ sequences; the following random variable:
 
\begin{equation}
\Xi_{n,q}(X^1,...,X^q)=\underset{l=0,...,n}\max\{X^1_{i+k}=X^2_{i+k}=...=X^q_{i+k} \text{ for } k=0,...,l \text{ and } 1\le i\le n-l \}.
\end{equation}

To make the connection between the previous sections, let us now consider, as in \cite{short}, the discrete symbolic dynamical system $(\mathcal{A}^{\mathbb{N}},\sigma,\mathbb{P})$, where $\sigma$ is the right-sided shift and $\mathbb{P}$ is the probability measure associated to the process. We consider the symbolic distance in $\mathcal{A}^\mathbb{N}$ defined by:

\begin{equation}\label{distance}
d(x^1,x^2)=\exp(-\inf\{i\ge0,\sigma^ix^1\ne \sigma^ix^2\}).
\end{equation}\\

For our purpose, we take $f=Id$. In this symbolic dynamics, the quantity $D_q$ (if it exists) identifies with a well-known quantity that we now define.

\begin{definition}
We call the R\'enyi entropy of order $q$ of $\mathbb{P}$, the following quantity (if the limit exists):

\begin{equation}\label{renyi}
H_q=\lim_{k\to\infty} \frac{\log\sum_{C_k} \mathbb{P}(C_k)^q}{(1-q)k},
\end{equation}\\

where $C_k(x)=\{y\in\mathcal{A}^{\mathbb{N}}: \sigma^ix=\sigma^iy\text{ for all } 0\le i\le k\}$ is the cylinder of length $k$ containing $x\in \mathcal{A}^{\mathbb{N}}$. 
\end{definition}

To see that $D_q$ identifies with $H_q$ in this context, it is enough to start from definition (\ref{defdq}), take $f=Id$ and use the symbolic distance, allowing to replace balls by cylinders.\\

The Dynamical Extremal Index $\theta_q=\theta_q^{Id}$ becomes in this set up (if it exists, and from equation (\ref{deftheta})):

\begin{equation}\label{thetaq}
\begin{aligned}
\theta_q&=1-p_{0,q}\\
        &=\lim_{k\to\infty} \mathbb{P}(\sigma^{k+1}x^1=\sigma^{k+1}x^2=...=\sigma^{k+1}x^q|\sigma^ix^1=\sigma^ix^2=...=\sigma^ix^q \text{ for } 0\le i \le k).
        \end{aligned}
\end{equation}

Indeed one sees easily that only the $p_{0,q}$ in definition (\ref{deftheta}) is non-zero in this situation (we provide a more detailed argument in the annex).\\

The quantity
\begin{equation}\label{phiq}
\begin{aligned}
Y_i&=-\log(\max_{s=2,\dots,q}d(x_i^1,x_i^s))\\
   &=\underset{j\ge 0}\inf\{\sigma^jx_i^1\ne \sigma^jx_i^s, \text{for some } s=2,...,q\}
   \end{aligned}
\end{equation}

is the length of the longest matching sub-sequence starting from the $i^{th}$ symbol of the different sequences. Now, the quantity

\begin{equation}\label{mn}
M_{n,q}(x^1,...,x^q)= \underset{i=0,n-1}{\max} Y_i
\end{equation}\\

is equal to $$\underset{l\in\mathbb{N}}\max\{x^1_{i+k}=x^2_{i+k}=...=x^q_{i+k} \text{ for } k=1,...,l \text{ and } 0\le i\le n-1 \}.$$ This object is closely related to the quantity $\Xi_{n,q}$ we are interested in. Since we work with different sequences of symbols, and $Y_i$ is a variable defined in the product space, we will state our results with respect to the product measure $\mathbb{P}_q$.  We prove our results under the hypothesis that the process has certain mixing properties, that we now recall.

\begin{definition}
The process $(\mathcal{A}^{\mathbb{N}},\sigma,\mathbb{P})$ is said to be $\alpha-$mixing if there exists $\alpha(n) \to 0$ such that 

\begin{equation}
\underset{A,B\subset \mathcal{A}^{\mathbb{N}}}\sup |\mathbb{P}(A \cap \sigma^{-n} B)- \mathbb{P}(A) \mathbb{P} (\sigma^{-n} B)| \le \alpha(n).
\end{equation}
\end{definition}

One could obtain a distributional result analogue to Proposition 1, by proving that conditions \foreignlanguage{russian}{Д}$_1(u_n)$ and  \foreignlanguage{russian}{Д}'$_1(u_n)$ are satisfied. With this approach, we get the following result, whose detailed proof can be found in the annex:  

\begin{result}
If the sequences are $\alpha$-mixing with $\alpha(n) <\beta e^{-\kappa n}$ for some $\beta\in \mathbb{R}^+$  and some $\kappa>H_q(q-1)$, and the limits defining $\theta_q$ and $H_q$ exist and are different from 0, then

  $$|\mathbb{P}_q(\Xi^q_n\le u_n(s))-\exp(-\theta_q\exp(-s))| \underset{n\to\infty}\to 0,$$ 
  with $u_n(s)=\lfloor\frac{\log n+s}{H_q(q-1)}\rfloor$.
\end{result}

\begin{remark}
We took $f=id$, to ensure a clustering structure that satisfies the different conditions of existence of the limit law, in particular condition \foreignlanguage{russian}{Д}'$_1(u_n)$. We could also consider, as in the first section of the paper, a non-trivial $f$. In the context of sequence matching, $f$ is called the encoding function (or encoder) and can model different treatments of the original source of information \cite{encoded}. The clustering structure is however in this case too complex to yield such a general result.
\end{remark}

It is in fact possible to obtain a more general result than Result 1, under much weaker conditions. The latter is based on results by Abadi and Saussol concerning the Hitting Time Statistics of symbolic dynamics in cylinders \cite{absa}. This idea originates from a discussion with J\'er\^ome Rousseau to whom the author is thankful.

\begin{theorem}
If $\mathbb{P}$ is $\alpha-$ mixing, and if the limits defining $\theta_q$ and $H_q$ exist and are different from 0, then

  $$|\mathbb{P}_q(\Xi_{n,q}\le u_n(s))-\exp(-\theta_q\exp(-s))| \underset{n\to\infty}\to 0,$$ 
  with $u_n=u_n(s)=\lfloor\frac{\log n+s}{H_q(q-1)}\rfloor$.
\end{theorem}

\begin{proof}
 Let us consider the process $(Z_i)$ defined by 
 
\beq
Z_i=
 \left\{
    \begin{array}{ll}
  1 \text{  if    } X_i^1=X_i^2=...=X_i^q,\\
  0 \text{    otherwise.}\\
    \end{array}
\right.
\nuq{17.5} 
 
The problem of finding the largest common substring to $X^1,...,X^q$ is now equivalent to find the longest succession of ones in the process $(Z_i)$. Let us consider the dynamical system $(\mathcal{B},\tilde{\mathbb{P}},\sigma)$, where $\mathcal{B}=\{0,1\}^{\mathbb{N}}$, $z$ a point in $\mathcal{B}$ and $\tilde{\mathbb{P}}$ the probability measure defined by

\begin{equation}
\begin{aligned}
\tilde{\mathbb{P}}(z_i=1)&=\mathbb{P}_q(x^1_i=...=x^q_i)\\
                         &=\sum_{a\in\mathcal{A}} \mathbb{P}(x^1_i=a)^q.
                         \end{aligned}
\end{equation}

Let us denote $I_k$ the cylinder constituted of all sequences having their first $k$ symbols equal to 1, and denote $$\tau_{I_k}(z) =\inf\{j \ge 1: \sigma^jz\in I_k\},$$

the first hitting time of the point $z$ in the set $I_k$. We notice that 

\begin{equation}\label{tt2}
\mathbb{P}_q(M_{n,q}<u_n)=\tilde{\mathbb{P}}(\tau_{I_{u_n}}>n).
\end{equation}

Since $\mathbb{P}$ is $\alpha-mixing$, so is $\tilde{\mathbb{P}}$, by theorem 5.1 in \cite{mixreview}. We are then in the set up of Theorem 1 in \cite{absa}. In particular, Hypothesis 1 of this theorem is satisfied, from Example 2 in \cite{absa}. Therefore:

\begin{equation}\label{tt1}
\sup_{t\in \mathbb{R}^+}|\tilde{\mathbb{P}}(\lambda(I_{u_n})\tilde{\mathbb{P}}(I_{u_n})\tau_{I_{u_n}} > t)-\exp(-t)| \underset{n\to\infty}\to 0,
\end{equation}

where, from \cite{absatheta}:

\begin{equation}
\begin{aligned}
\lambda(I_{u_n}) &=1-\lim_{k\to\infty} \frac{\tilde{\mathbb{P}}(I_{k+1})}{\tilde{\mathbb{P}}(I_k)}\\
                 &=1-\lim_{k\to\infty} \frac{\sum_{C_{k+1}}\mathbb{P}(C_{k+1})^q}{\sum_{C_k}\mathbb{P}(C_k)^q}\\
                 &= \theta_q.
\end{aligned}
\end{equation}

Notice now that we have from equation (\ref{renyi}):

\begin{equation}
\begin{aligned}
H_q&=\lim_{k\to\infty} \frac1{(1-q)k}{\log \underset{C_k}\sum \mathbb{P}(C_k)^q}\\
   &=\lim_{k\to\infty} \frac{\log \tilde{\mathbb{P}}(I_k)}{(1-q)k},
   \end{aligned}
\end{equation}

so that 

\begin{equation}
\tilde{\mathbb{P}}(I_{u_n})\underset{n\to\infty}\sim e^{-(q-1)H_qu_n}.
\end{equation}

 If we put $t=e^{-s}$, equation (\ref{tt1}) writes, after rearranging a bit:

\begin{equation}
\sup_{s\in \mathbb{R}} |\tilde{\mathbb{P}}(\tau_{I_{u_n}} > e^{-s+(q-1)H_qu_n})-\exp(-\theta_q e^{-s})| \underset{n\to\infty}\to 0.
\end{equation}

keeping in mind that $u_n=\lfloor\frac{\log n+s}{H_q(q-1)}\rfloor$, we get

\begin{equation}
\sup_{s\in \mathbb{R}} |\tilde{\mathbb{P}}(\tau_{I_{u_n}} > n)-\exp(-\theta_q e^{-s})| \underset{n\to\infty}\to 0.
\end{equation}

Using now equation (\ref{tt2}), we obtain that for all $s \in \mathbb{R}$:

\begin{equation}\label{evl}
 |\mathbb{P}_q(M_{n,q} > u_n)-\exp(-\theta_q e^{-s})| \underset{n\to\infty}\to 0.
\end{equation}

Now that we have a distributional result for the variable $M_{n,q}$, we can get one for $\Xi_{n,q}$, which is a slightly different object. In fact we have that

\begin{equation}\label{kkk}
\begin{split}
\mathbb{P}_q(\Xi_{n,q}\le u_n)&=\mathbb{P}_q(\Xi_{n,q}\le u_n \cap M_{u_n,q}(\sigma^{n-u_n}x^1,...,\sigma^{n-u_n}x^q)\le u_n)\\
                                                                &+\mathbb{P}_q(\Xi_{n,q}\le u_n \cap M_{u_n,q}(\sigma^{n-u_n}x^1,...,\sigma^{n-u_n}x^q) > u_n).
\end{split}
\end{equation}

The second term is bounded above by the term $\mathbb{P}_q(M_{u_n,q}(\sigma^{n-u_n}x^1,...,\sigma^{n-u_n}x^q)> u_n)$, which, by invariance of the measure by $\sigma$ equals $\mathbb{P}_q(M_{u_n,q}(x^1,...x^n)> u_n)$, which is clearly vanishing to 0 as $n\to \infty$, from (\ref{evl}).\\

The first term in (\ref{kkk}) is exactly equal to $\mathbb{P}_q(M_{n,q}(\overline{x})\le u_n)$. Therefore:

\begin{equation}\label{rr}
|\mathbb{P}_q(\Xi_{n,q}\le u_n)-\mathbb{P}_q(M_{n,q}\le u_n)| \underset{n\to\infty}\to 0.
\end{equation}

We have that for all $s\in \mathbb{R}$:

\begin{equation}
|\mathbb{P}_q(\Xi_{n,q}\le u_n)-\exp(-\theta_qe^{-s})| \le |\mathbb{P}_q(\Xi_{n,q}\le u_n)-\mathbb{P}_q(M_{n,q}\le u_n)|+ |\mathbb{P}_q(M_{n,q} \le  u_n) -\exp(-\theta_qe^{-s})|,
 \end{equation}
which, by relations (\ref{evl}) and (\ref{rr}) goes to 0.
\end{proof}


\section{Acknowledgement}
The author was partially supported by CMUP, which is financed by national funds through FCT – Fundação para a
Ciência e Tecnologia, I.P., under the project with reference UIDB/00144/2020.
The author would like to thank Jorge M. Freitas, J\'er\^ome Rousseau, Beno\^it Saussol and Sandro Vaienti for the fruitful discussions we had concerning this work and the anonymous referee for its constructive comments.

\section{Annex (proof of Result 1, via EVT)}

We first show that both conditions \foreignlanguage{russian}{Д}$_1(u_n)$ and  \foreignlanguage{russian}{Д}'$_1(u_n)$ are satisfied, so we have an EVL for the random variable $M_{n,q}$. Then we show that $\Xi_{n,q}$ and $M_{n,q}$ have the same asymptotic distribution. Let us first take care of condition  \foreignlanguage{russian}{Д}$'_1(u_n)$.\\
We observe that if $Y_0=k\in \mathbb{N}^*$, then  $Y_j=k-j$ for $1\le j\le k$. Therefore, if $Y_0>u_n$, then $Y_1>Y_j>u_n-j$ for $2\le j<u_n$, so that all the probabilities in the sum in point 3 of definition 4 are 0 for $2\le j<u_n$, that is

\begin{equation}\label{11}
\underset{n\to\infty}\lim n\sum_{j=2}^{u_n-1}\mathbb{P}_q(Y_0>u_n\cap Y_1\le u_n\cap Y_{j}>u_n)=0.
\end{equation}

Let $0<\varepsilon_2<\varepsilon_1<1$ and $C_1=1-\varepsilon_1$. We define $r_n=\lfloor C_1u_n\rfloor$ and $\lambda_n=\lfloor n^{\varepsilon_2}\rfloor$. We take $j$ such that $u_n\le j \le \lambda_n$. We observe that 
$\{Y_j>u_n\} \subset \{Y_{j+r_n}>u_n-r_n\}$, so that 

\begin{equation}\label{h}
 \mathbb{P}_q(Y_0>u_n\cap Y_1\le u_n\cap Y_{j}>u_n) \le  \mathbb{P}_q(Y_0>u_n\cap Y_1\le u_n\cap Y_{j+r_n}>u_n-r_n).
\end{equation}                                                                          

Notice that $\{Y_0>u_n\cap Y_1\le u_n\}=\{Y_0=u_n+1\}$, and this event depends only on the realizations of $\bar{X}_1,...,\bar{X}_{u_n+2}$, whereas $\{Y_{j+r_n}>u_n-r_n\}$ depends only on the realizations of $\bar{X}_{j+r_n},\bar{X}_{j+r_n+1},...$, which puts a gap of length $j+r_n-u_n-2$. We now use the fact that the sequences are $\alpha-$mixing, which implies that the $q-$fold Cartesian product of the sequences is $(\alpha_q)$-mixing, with $\alpha_q(n)\le q\alpha(n)$ (see theorem 5.1 in \cite{mixreview}). We have 

\begin{equation}
\begin{aligned}
\mathbb{P}_q(Y_0>u_n\cap Y_1\le u_n\cap Y_{j}>u_n) &\le \alpha_q(j+r_n-u_n-2)+\mathbb{P}_q(Y_0>u_n\cap Y_1 \le u_n)\mathbb{P}_q(Y_{j+r_n}>u_n-r_n)\\
                                                                              &\le q\alpha(j+r_n-u_n-2)+ \mathbb{P}_q(Y_0>u_n)\mathbb{P}_q(Y_{j+r_n}>u_n-r_n)\\
                                                                              &\le q\beta e^{-\kappa(j+r_n-u_n-2)} + \mathbb{P}_q(Y_0>u_n)\mathbb{P}_q(Y_0>u_n-r_n).
\end{aligned}
\end{equation}
                                                                                   
To get the last inequality, we used the invariance of the measure. Notice that $j\ge u_n$, so that $e^{-\kappa(j+r_n-u_n-2)}\le e^{-\kappa(r_n-2)}$.
We also have from relation (\ref{2}) that  $\mathbb{P}_q(Y_0>u_n) \sim e^{-u_n\tau_q}$, where $\tau_q=H_q(q-1)$, so that there exists $C_2>1$ such that $$\mathbb{P}_q(Y_0>u_n)<C_2e^{-u_n\tau_q}.$$

We then have:

\begin{equation}
\mathbb{P}_q(Y_0>u_n\cap Y_1\le u_n\cap Y_{j}>u_n) \le q\beta e^{-\kappa(r_n-2)}+C_2^2e^{-(2u_n-r_n)\tau_q}.
\end{equation}

Then we can write 
\begin{equation}\label{ll}
\begin{aligned}
n\sum_{j=u_n}^{\lambda_n}\mathbb{P}_q(Y_0>u_n\cap Y_1\le u_n\cap Y_{j}>u_n) &\le \sum_{j=u_n}^{\lambda_n} [n q\beta e^{-\kappa(r_n-2)}+nC_2^2e^{-(2u_n-r_n)\tau_q}]\\
                                                                                                                                                    &\le (\lambda_n -u_n)n q\beta e^{-\kappa (r_n-2)}+(\lambda_n-u_n)nC_2^2e^{-(2u_n-r_n)\tau_q}\\
                                                                                                                                                    &\le \lambda_n n q\beta e^{-\kappa (r_n-2)}+\lambda_n nC_2^2e^{-(2u_n-r_n)\tau_q}\\
                                                                                                                                                    &\le (q\beta e^{2\kappa})n\lambda_n e^{-\kappa r_n}+C_2^2n\lambda_ne^{-2(u_n-r_n)\tau_q}\\
                                                                                                                                                    &\le (q\beta e^{2\kappa})n\lambda_n e^{-\kappa \lfloor C_1 u_n \rfloor}+C_2^2n\lambda_ne^{-2(u_n-\lfloor C_1 u_n \rfloor)\tau_q}\\
                                                                                                                                                     &\le (q\beta e^{2\kappa})n\lambda_n e^{-\kappa (C_1 u_n-1)}+C_2^2n\lambda_ne^{-(2-C_1)u_n \tau_q}\\
                                                                                                                                                     &\le (q\beta e^{3\kappa})n\lambda_ne^{-\kappa C_1 \lfloor \frac{\log n+s}{\tau_q}\rfloor}+C_2^2n\lambda_ne^{-(2-C_1)\lfloor \frac{\log n+s}{\tau_q}\rfloor \tau_q}\\
                                                                                                                                                     &\le (q\beta e^{3\kappa})n\lambda_n e^{-\kappa C_1 ( \frac{\log n+s}{\tau_q}-1)}+C_2^2n\lambda_ne^{-(2-C_1)( \frac{\log n+s}{\tau_q}-1) \tau_q}\\                                                                                                                                                    
                                                                                                                                                     &\le C_3n\lambda_n e^{-\kappa C_1 \frac{\log n}{\tau_q}}+C_4n\lambda_n e^{-(2-C_1)\log n},\\                                                                                                                                                    
\end{aligned}
\end{equation}

with

$$C_3=q\beta e^{3\kappa}e^{-\kappa C_1(\frac{s}{\tau_q}-1)}$$

and 

$$C_4=C_2^2e^{(C_1-2)(s-\tau_q)}.$$

For the first term, we have 

\begin{equation}\label{78}
\begin{aligned}
C_3n\lambda_n e^{-\kappa C_1 \frac{\log n}{\tau_q}}&=C_3n\lfloor n^{\varepsilon_2}\rfloor e^{-\kappa C_1 \frac{\log n}{\tau_q}}\\
                                                                                     & \le C_3n n^{\varepsilon_2} e^{-\kappa C_1 \frac{\log n}{\tau_q}}\\
                                                                                     & \le C_3 n^{1+\varepsilon_2- \frac{\kappa C_1}{\tau_q}}\\
                                                                                     &\le C_3 n^{1+\varepsilon_2- \frac{\kappa (1-\varepsilon_1)}{\tau_q}}.\\
\end{aligned}
\end{equation}

Since $\kappa > \tau_q$, we can always chose $\varepsilon_1,\varepsilon_2$ and $\varepsilon_3 >0$ such that 

\begin{equation}
\epsilon_3>\frac{(\varepsilon_1+\varepsilon_2)\tau_q}{1-\varepsilon_1}
\end{equation}

and

\begin{equation}
\kappa >\tau_q+\varepsilon_3.
\end{equation}
 
We then have 

\begin{equation}
\begin{aligned}
1+\varepsilon_2- \frac{\kappa (1-\varepsilon_1)}{\tau_q}&<\varepsilon_1+\varepsilon_2-\frac{\varepsilon_3(1-\varepsilon_1)}{\tau_q} <0.
\end{aligned}
\end{equation}

and so by relation (\ref{78}):

\begin{equation}\label{44}
C_3n\lambda_n e^{-\kappa C_1 \frac{\log n}{\tau_q}}  \underset{n\to\infty}\to 0.
\end{equation}

Let us now come to the second term in relation (\ref{ll}):

\begin{equation}
\begin{aligned}
C_4n\lambda_n e^{-(2-C_1)\log n}&=C_4n\lfloor n^{\varepsilon_2} \rfloor e^{-(1+\varepsilon_1)\log n}\\
                                                   &\le C_4 n^{\varepsilon_2-\varepsilon_1}.\\
\end{aligned}
\end{equation}

And since $\varepsilon_1>\epsilon_2$:

\begin{equation}\label{45}
C_4n\lambda_n e^{-(2-C_1)\log n}  \underset{n\to\infty}\to 0.
\end{equation}

Combining relations (\ref{44}), (\ref{45}) and (\ref{ll}), we have that:

\begin{equation}\label{22}
\underset{n\to\infty}\lim n\sum_{j=u_n}^{\lambda_n}\mathbb{P}_q(Y_0>u_n\cap Y_1\le u_n\cap Y_{j}>u_n)=0.
\end{equation}

Combining equation (\ref{11}) and (\ref{22}), we get that 

\begin{equation}\label{23}
\underset{n\to\infty}\lim n\sum_{j=2}^{\lambda_n}\mathbb{P}_q(Y_0>u_n\cap Y_1\le u_n\cap Y_{j}>u_n)=0.
\end{equation}

Taking $k_n=\frac{n}{\lambda_n}=n^{1-\varepsilon_2}$, we have that points 1 and 3 of condition \foreignlanguage{russian}{Д}$_1'(u_n)$ are satisfied. To satisfy point 2, we take $t_n=\lfloor n^{\varepsilon_4}\rfloor$, with $\varepsilon_4<\varepsilon_2$.  \foreignlanguage{russian}{Д}$_1'(u_n)$ is then satisfied.\\

Let us now come to condition  \foreignlanguage{russian}{Д}$_1(u_n)$. Define the event $\Omega_n = \{Y_0<\lfloor t_n/2 \rfloor\}$. We have that

\begin{equation}\label{jjjj}
\begin{aligned}
 |\mathbb{P}_q(A_n\cap B_{t_n,l,n})-\mathbb{P}_q(A_n)\mathbb{P}_q(B_{0,l,n})| &\le  |\mathbb{P}_q(A_n\cap \Omega_n\cap B_{t_n,l,n})-\mathbb{P}_q(A_n\cap\Omega_n)\mathbb{P}_q(B_{0,l,n})|\\
                                                                                                                                                                                              &+|\mathbb{P}_q(A_n\cap \Omega_n^c\cap B_{t_n,l,n})-\mathbb{P}_q(A_n\cap\Omega_n^c)\mathbb{P}_q(B_{0,l,n})|.
 \end{aligned}
 \end{equation}

We have just introduced a gap of size $t_n-\lfloor t_n/2\rfloor-1$ in the first term of the right hand side of the previous inequation. Indeed, for $n$ large enough, the event $A_n\cap \Omega_n$ depends only on the realizations of $\bar{X}_1,\bar{X}_2,...,\bar{X}_{\lfloor t_n/2\rfloor}$, while $B_{t_n,l,n}$ depends on the realizations of $\bar{X}_{t_n},...,\bar{X}_{t_n+l}$. We can then bound this term, using again theorem 5.1 in \cite{mixreview} :

\begin{equation}\label{hulk1}
\begin{aligned}
 |\mathbb{P}_q(A_n\cap \Omega_n\cap B_{t_n,l,n})-\mathbb{P}_q(A_n\cap\Omega_n)\mathbb{P}(B_{0,l,n})|&\le q\alpha(t_n-\lfloor t_n/2\rfloor-1)\\
                                                                                                                                                         &\le q\beta e^{-\kappa(t_n-\lfloor t_n/2\rfloor-1)}\\
                                                                                                                                                         &\le q\beta e^{-\kappa(t_n/2-1)}.\\
\end{aligned}
\end{equation}
 
 For the second term, we can write
 
 \begin{equation}\label{hulk2}
 \begin{aligned}
 |\mathbb{P}_q(A_n\cap \Omega_n^c\cap B_{t,l,n})-\mathbb{P}_q(A_n\cap\Omega_n^c)\mathbb{P}_q(B_{0,l,n})|& \le |\mathbb{P}_q(A_n\cap \Omega_n^c\cap B_{t,l,n})|+|\mathbb{P}_q(A_n\cap\Omega_n^c)\mathbb{P}_q(B_{0,l,n})|\\
                                                                                                                                                             & \le 2\mathbb{P}_q(\Omega_n^c) \sim 2e^{-\lfloor t_n/2 \rfloor \tau_q}\\
                                                                                                                                                             &\le C_5e^{-\lfloor t_n/2 \rfloor\tau_q}\\
                                                                                                                                                             &\le  C_5e^{- (t_n/2-1) \tau_q},\\
\end{aligned}
\end{equation}
 for some $C_5>2$.\\
 
 Let us now take $$\gamma(n,t_n)=q\beta e^{-\kappa(t_n/2-1)}+C_5e^{-(t_n/2-1) \tau_q}.$$
 
 Combining expressions (\ref{jjjj}), (\ref{hulk1}), (\ref{hulk2}), we get
 
 \begin{equation}\label{ffcond2}
 |\mathbb{P}_q(A_n\cap B_{t,l,n})-\mathbb{P}_q(A_n)\mathbb{P}_q(B_{0,l,n})| \le \gamma(n,t_n).
 \end{equation}

  Let us recall that from condition \foreignlanguage{russian}{Д}$_1'(u_n)$, $t_n=\lfloor n^{\varepsilon_4}\rfloor=o(n).$ $\gamma$ is clearly decreasing and we check easily that $$n\gamma(n,t_n) \underset{n\to\infty}\to 0.$$
  
  Condition \foreignlanguage{russian}{Д}$_1(u_n)$ is then satisfied. We can now apply corollary 4.1.7 in \cite{book} to get that
  
 \begin{equation}
\mathbb{P}_q(M_{n,q} \le u_n) -\exp(-\theta_q\exp(-s)) \underset{n\to\infty}\to 0.
\end{equation}

We conclude by using the same arguments as in the proof of Theorem 1, showing that $M_{n,q}$ and $\Xi_{n,q}$ have the same limit distribution.

\end{document}